\newcommand{\mbc}{{\mathbb C}}

\newcommand{\clf}{{\mathbb C}{\rm l}}

\newcommand{\mbp}{{\mathbb P}}

\newcommand{\mbz}{{\mathbb Z}}

\newcommand{\ra}{{\rangle}}

\documentclass[11pt,draft,reqno]{amsart}
     \makeatletter
     \def\section{\@startsection{section}{1}%
     \z@{.7\linespacing\@plus\linespacing}{.5\linespacing}%
     {\bfseries
     \centering
     }}
     \def\@secnumfont{\bfseries}
     \makeatother
\usepackage{amssymb}
\usepackage[final]{hyperref}
\usepackage{tikz}

\newtheorem{theorem}{Theorem}[section]

\newtheorem{prop}[theorem]{Proposition}

\theoremstyle{definition}

\theoremstyle{remark}

\numberwithin{equation}{section}
 
\begin{document}

\title[Finite Geometries with Qubit  Operators]{Finite Geometries with Qubit  Operators}

\date{3 April, 2009} 

  \author{Ambar N. Sengupta}

\address{Ambar N. Sengupta, Department of Mathematics, Louisiana State University, Baton Rouge, LA 70803, USA}
 \email{sengupta@math.lsu.edu}
\urladdr{\url{http://www.math.lsu.edu/~sengupta}}



\thanks{Research supported by US National Science Foundation Grant DMS-0601141}

\begin{abstract} Finite projective geometries, especially the Fano plane, have been observed to arise in the
context of certain quantum gate
operators. We use Clifford algebras to explain why these geometries, both planar and higher dimensional, 
 appear in the context of multi-qubit composite systems.
\end{abstract}

\maketitle

\section{Introduction}\label{S:Intro}

Finite projective geometries have appeared in several investigations relating to quantum computing; 
for example, Levay et al. \cite{LSV},   Rau \cite{Rau}, and Saniga et al. \cite{San}.
In this paper we present a mathematical  explanation  for the appearance of these geometries. We show that
finite geometries arise from  units in Clifford algebras which form a vector space over the field $\mbz_2$ of two elements.

As noted before there are several works where finite geometries have been observed to arise in the context of
quantum computing. The online resource of Saniga \cite{SanWeb} lists many such works. Intriguing relations
between finite geometries and the entropy of black hole solutions in certain supersymmetric field theories are explored in
L\'evay et al. \cite{LSV} (and other works cited therein).
Havlicek \cite{Hav}
has also discussed a mathematical explanation for the appearance of these geometries, but does not use Clifford algebras. Our work
is closest to and inspired by the work of Rau \cite{Rau}.

\section{Appearance of Finite Geometries}\label{S:FinG}
 
First let us briefly recall that a {\em projective geometry} is specified by a set ${\mathcal P}$ of {\em points}, a set ${\mathcal L}$ of {\em lines},
an incidence relation ${\mathcal I}\subset {\mathcal P}\times {\mathcal L}$ (for which we say that a point $A$ lies on a line $l$, or that $l$
passes through $A$, if
$(A,l)\in {\mathcal I}$) such that the following hold: (i) for every pair of distinct points  $A,B\in {\mathcal P}$ there is a unique line
denoted  $AB$  which passes through  $A$  and  $B$; (ii) if $A,B,C,D$ are   points such that the lines $AB$ and $CD$ have a point in common, then the lines 
$AC$ and $BD$ also have a point in common;
(iii) every line passes through at least three points.   A set of points  are {\em collinear} if they lie on a  common line; a set of lines are {\em coincident} if they pass through a common point.
A {\em triangle} is simply a set of three distinct points.

Under additional geometric hypotheses, the  axioms (i)-(iii) can be used to construct a field $k$
and a vector space $V$ over $k$ such that points correspond to one-dimensional subspaces of $V$, lines to two-dimensional subspaces of $V$, and
incidence corresponds to the subspace relation.  For any vector space $V$, this specifies a projective geometry ${\mathbb P}(V)$, which has a special
feature called the {\em Desargues property}:  if  $ABC$ and $A'B'C'$ are triangles  such that there is a point $D$ which, for each $X\in\{A,B,C\}$, is collinear with
$X$ and $X'$, then there is a line $l$ which, for every pair of distinct points $X,Y\in\{A,B,C\}$, is coincident with the lines $XY$ and $X'Y'$. This is illustrated in Figure \ref{FigDes}.

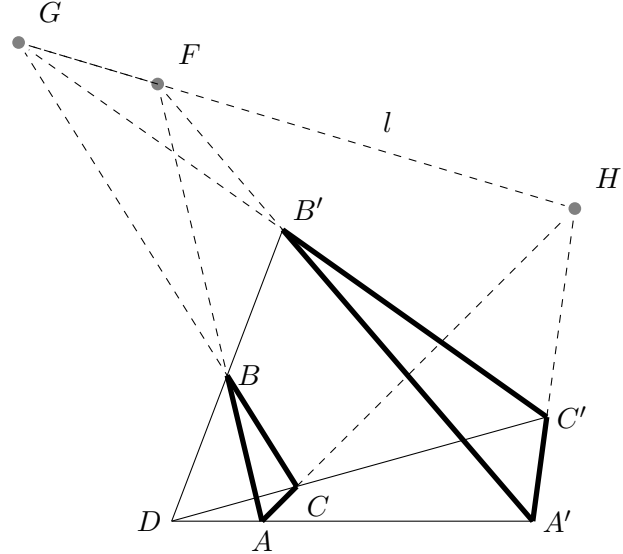
\begin{figure}

\begin{tikzpicture}[scale=1.2]

\coordinate [label=left: {$D$}] (V) at (0,0,0);

\coordinate [label=below: {$A$}] (A) at (1,0,0);
\coordinate [label=right: {$B$ }] (B) at (1,2,1);
\coordinate [label=-45: {$C$ }] (C) at (1,0,-1);

\coordinate [label=right: {$A'$}] (A') at (4,0,0);
\coordinate [label=45: {$B'$ }] (B') at (2,4,2);
\coordinate [label=right: {$C'$ }] (C') at (3,0,-3);

\coordinate [label=right: {$l$}] (l) at (3,5.2,2);

\draw (A)--(V);
\draw (B)--(V);
\draw (C)--(V);

\draw (A)--(A');
\draw (B)--(B');
\draw (C)--(C');

\draw [line width =2pt] (A)--(B);
\draw [line width =2pt] (B)--(C);
\draw [line width =2pt] (C)--(A);

\draw [line width =2pt] (A')--(B');
\draw [line width =2pt] (B')--(C');
\draw [line width =2pt]  (C')--(A');

\node [label= 45:$F$]   (F) at (intersection of A--B and A'--B') {};

\draw [dashed] (A)--(F); \draw [dashed] (A')--(F);

\node [label= 45:$G$]   (G) at (intersection of B--C and B'--C') {};

\draw [dashed] (C)--(G); \draw [dashed] (C')--(G);

 \node [label= 45:$H$]   (H) at (intersection of A--C and A'--C') {};

 \draw [dashed] (C)--(H); \draw [dashed] (C')--(H);

 \draw [dashed] (F)--(G); \draw [dashed] (G)--(H);
 \foreach \point in {F,G,H}
\fill [black,opacity=.5] (\point) circle (2pt);
  \end{tikzpicture}
\caption{The Desargues Configuration}
\label{FigDes}
\end{figure}

Returning to ideas in quantum computing, most basic quantum gates are constructed from two-state systems, i.e. those with two-dimensional Hilbert spaces. The composite of such systems is
described by tensor products of the corresponding Hilbert spaces. On choosing a fixed orthonormal basis $|0\ra$ and $|1\ra$ in the two-dimensional
Hilbert space, we can identify this space with $\mbc^2$. Some of the basic quantum gates are expressed using tensor products of
the {\em Pauli matrices}
\begin{equation} \label{Paulim}
\sigma_1=\sigma_x=\left[\begin{matrix} 0&1\\1& 0\end{matrix}\right],\qquad \sigma_2=\sigma_y=\left[\begin{matrix} 0& -i\\i& 0\end{matrix}\right],\qquad\sigma_3=\sigma_z=\left[\begin{matrix} 1&0\\0& -1\end{matrix}\right]
\end{equation}
 along with the identity matrix  $\sigma_0=I$. 
The Pauli matrices have zero trace, are hermitian as well as unitary, and form a basis of the three-dimensional real vector space of
all zero-trace $2\times 2$ hermitian matrices; along with $\sigma_0$, their real-linear span is the space of all hermitian matrices, and their
complex linear span is the space of all $2\times 2$ complex matrices. 

Let $N$ be a fixed positive integer; we will work with gates which process $N$ qubits, i.e. systems with state Hilbert space being
$(\mbc^2)^{\otimes N}$. Let $S$ be the set of all   products of the  $N$-fold tensor products of the matrices $\sigma_{\alpha}$ with $\alpha\in\{0,1,2,3\}$.

Consider now a geometry constructed as follows. Points are the elements of $S$, other than just the identity $I$ and also $-I$,  with $x$ and $-x$ identified. Lines contain three points, with
 points $a,b,c$ being on a line if $ab$ equals $\pm c$.   Rau \cite{Rau} shows (among other results) that
 for $N\in\{1,2\}$, this results in a projective space, and, furthermore,   the Desargues property holds in this space.
 A variety of other such observations have been made by Levay et al. \cite{LSV}.
 
 In particular, consider the special case $N=1$. In this case, there are seven points, three corresponding to the Pauli matrices $\sigma_j$,
 the three pairwise products $\sigma_j\sigma_k$, with $j\neq k$, and the triple product $\sigma_1\sigma_2\sigma_3$. Thus, one can view this projective space
 in terms of a triangle with the Pauli matrices as vertices, with the point corresponding to  $\sigma_j\sigma_k$ (with $j\neq k$) lying
 on the line joining the vertices for $\sigma_j$ and $\sigma_k$, and a centroid-style point in the interior of the triangle, corresponding to
 $\sigma_1\sigma_2\sigma_3$. This space is called the {\em Fano plane} (see Figure  \ref{FigFano}).  As Rau \cite{Rau} points out, it also appears as a subspace of the projective space associated to
 the $N=2$ system. The appearance of the Fano plane in quantum computing related contexts has been noted
 in many works, including Rau \cite{Rau}, Levay \cite{LSV} (and see the many works listed at   \cite{SanWeb}).

\section{Finite Geometries from Clifford Algebras}\label{S:FinClif}

In this section we shall explain why the finite projective geometries, satisfying the Desargues property, described in the preceding section arise.  

The usual Clifford algebra $\clf_m$ over $m$ generators $e_1,...,e_m$ is the associative complex algebra with identity generated by these elements, along with
an identity element $I$, subject to the relations
$$e_je_k+e_ke_j=2\delta_{jk}I\qquad\hbox{for all $j,k\in\{1,...,m\}$.}$$
Note that the square of each $e_j$ is $I$.  This algebra (for a construction see Artin \cite{Artin}) has dimension $2^m$, and a basis is formed by the products 
$$e_S=e_{s_1}...e_{s_r},$$
for all subsets $S=\{s_1,...,s_r\}$, with $s_1<s_2<\cdots<s_r$, of $\{1,...,m\}$; the element $I=e_{\emptyset}$ is the multiplicative identity.  The product
is given by
$$e_Se_T=\epsilon_{S,T}e_{S\Delta T},$$
where $S\Delta T$ is the set of elements in $S\cup T$ which are not in $S\cap T$, and $\epsilon_{S,T}$ is the product of all $\epsilon_{st}$, with $s\in S$ and
$t\in T$, where $\epsilon_{st}=-$ if $s>t$ and $\epsilon_{st}=1$ if $s\leq t$.

The Clifford algebra is a superalgebra, splitting into a sum of even and odd elements:
$$\clf_m={\clf}_m^{0}\oplus {\clf}_m^{1},$$
where ${\clf}_m^{0}$ is spanned by products of even numbers of the elements $e_j$, and ${\clf}_m^{1}$ spanned by products of
odd numbers of elements.  It is important to note that here we are working with the Clifford algebra, not a representation of it; for example,
with $m=3$, if we consider the representation of $\clf_3$ using Pauli matrices, $\sigma_1\sigma_2$ is equal to $i\sigma_3$ and the splitting
into even and odd elements is not meaningful.

It is a standard and readily verified fact that
\begin{equation}\label{cliftens}
{\clf}_{n+m}\simeq {\clf}_n\otimes {\clf}_m,\end{equation}
where the tensor product of algebras is in the `super' sense, i.e. with multiplication specified through
$$(x\otimes y)(z\otimes w)=(-1)^{pq}(xz)\otimes (yw),$$
where $y\in {\clf}_m^p$  and $z\in {\clf}_m^q$, and $p,q\in\{0,1\}$.

When $m$ is odd, the Clifford algebra ${\clf}_m$ has exactly two distinct irreducible representations, each of dimension $2^{(m-1)/2}$, and the
image of the (complex) Clifford algebra, in each case, is the entire matrix algebra in the representation space.
For example, for $m=3$, the Pauli matrices give rise to one such representation, with the element $e_j\in {\clf}_3$ represented by the
matrix $\sigma_j$, for $j\in\{1,2,3\}$.

As noted before, many of the basic quantum gates are described by means of tensor products of the matrices $\sigma_{\alpha}$. Since $\otimes^n{\clf}_3$
is isomorphic to ${\clf}_{3n}$, we may as well view the quantum gate operators as elements of  a Clifford algebra  ${\clf}_m$, represented
on some finite-dimensional Hilbert space. 

Now consider the set $C_m$ consisting of all elements in ${\clf}_m$ of the form $\pm e_S$, for $S\subset\{1,...,m\}$. Observe that $C_m$ is a group under multiplication (often called the Clifford group),
and, furthermore, the square of each element of $C_m$ is $\pm I$. Thus, in the quotient group 
$$V_m=C_m/\{I,-I\}$$
 the square of every element is the identity
and so, in particular, $V_m$ is abelian. It will be convenient to write the group operation in the abelian group $C_m$ {\em additively}, even if it is notationally
somewhat counter-intuitive. Then  we would write the identity element as $0$. Because the square of each element is the identity, we have
$$1.x+1.x=x+x=0\qquad\hbox{for all $x\in V_m$.}$$
This ensures that $V_m$ is a vector space over the two-element field
$$\mbz_2=\mbz/2\mbz=\{0,1\}.$$
Let us now consider the {\em projective space} $\mbp(V_m)$ for $V_m$. The {\em points} of $\mbp(V_m)$ are the one-dimensional subspaces of $V_m$. The {\em lines} in $\mbp(V_m)$
are two-dimensional subspaces of $C_m$. 

For any $\pm e_S\in C_m$, with $S$ a  {\em non-empty} subset of $\{1,...,m\}$,  we then have a point in $\mbp(V_m)$; and, conversely,
every point of $\mbp(V_m)$ is a subspace of the form $\{0,\pm e_S\}\subset V_m$, for some non-empty $S\subset\{1,...,m\}$. A line in $\mbp(V_m)$, being a two-dimensional
subspace of $V$, is spanned by two vectors in $V_m$, i.e. it is a subset of $V_m$ of the form $\{0,\pm e_S, \pm e_T, \pm e_{S\Delta T}\}$, for two distinct
non-empty subsets $S$ and $T$ of $\{1,...,m\}$.  Thus, a line in the projective space $\mbp(V_m)$ is determined by two elements $e_S$ and $e_T$ and contains
also a third element corresponding to the product $e_Se_T$.

Thus, we see that the geometry formed by taking as points the elements $\pm e_S$ (one point for the pair $e_S, -e_S$), and taking three points $a,b,c$ to be collinear if the product of two is the third,
is precisely the geometry of the projective space $\mbp(V_m)$. 

For $m=3$ this yields the Fano plane which has seven points and seven lines. This is illustrated in   Figure \ref{FigFano}, wherein the circular path inside the triangle is one of the seven lines.

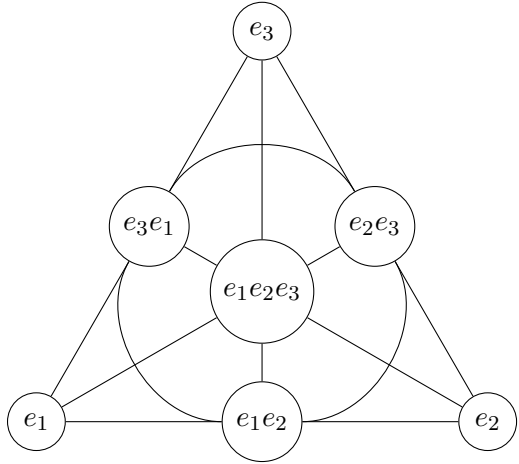
\begin{figure}

\begin{tikzpicture}[scale=3] 
\coordinate [label=left: {}] (A) at (-1,0);
\coordinate [label=right: { }] (B) at (1,0);
\coordinate [label=right: { }] (C) at (0,1.73205);

\coordinate [label=right: { }] (a) at (0.5,0.866);
\coordinate [label=left: { }] (b) at (-0.5,0.866);
\coordinate [label=below: { }] (c) at (0,0);
\coordinate [label=below: { }] (D) at (intersection of C--c and A--a);

\path 
(A)  node (nodeA) [shape=circle,draw] {$e_1$}

(B) node  (nodeB) [shape=circle,draw] {$e_2$}

(C) node (nodeC) [shape=circle,draw] {$e_3$}

(c) node (nodec)  [shape=circle,draw] {$e_1e_2$}

(a) node (nodea) [shape=circle,draw] {$e_2e_3$}

 (b) node (nodeb) [shape=circle,draw] {$e_3e_1$}
 
 (D) node (nodeD) [shape=circle,draw] {$e_1e_2e_3$}

;
\draw [] (nodeA) -- (nodec); 
\draw[ ] (nodec) -- (nodeB);
\draw[ ] (nodeA) -- (nodeb);
\draw[ ] (nodeb) -- (nodeC);
\draw[ ] (nodeC) -- (nodea);
\draw[ ](nodea) -- (nodeB);
\draw[ ] (nodeA) -- (nodeD);
\draw[ ] (nodeB) -- (nodeD);
\draw[ ] (nodeC) -- (nodeD);
\draw[ ] (nodea) -- (nodeD);
\draw[ ] (nodeb) -- (nodeD);
\draw[ ] (nodec) -- (nodeD);

\draw [] (nodea) to [out=120,in=60] (nodeb);

\draw [] (nodeb) to [out=240,in=180] (nodec);

\draw [] (nodec) to [out=0,in=-60] (nodea);
 \end{tikzpicture}
\caption{The Fano Plane}
\label{FigFano}
\end{figure}

It is a general fact  that the projective space $\mbp(V_m)$ (of any vector space $V_m$) satisfies the Desargues property. This explains the observations
made by Rau \cite[Figure 3]{Rau}. An illustration, in terms of Pauli matrices, is in Figure \ref{FigDesPauli}, where we work with the Clifford algebra
${\clf}_6\simeq{\clf}_3\otimes{\clf}_3$ (super-tensor product) and the vertices are labeled up to sign, and we have used the Pauli matrices for a specific representation of
  ${\clf}_3$.

\begin{figure}

\begin{tikzpicture}[scale=1.2]

\coordinate [label=left: {$\sigma_1\sigma_2\otimes I$}] (V) at (0,0,0);

\coordinate [label=below: {$\sigma_1\otimes\sigma_1$}] (A) at (1,0,0);
\coordinate [label=right: {$\sigma_1\otimes\sigma_2$ }] (B) at (1,2,1);
\coordinate [label=-45: {$\sigma_1\otimes\sigma_3$ }] (C) at (1,0,-1);

\coordinate [label=right: {$\sigma_2\otimes\sigma_1$}] (A') at (4,0,0);
\coordinate [label=45: {$\sigma_2\otimes\sigma_2$ }] (B') at (2,4,2);
\coordinate [label=right: {$\sigma_2\otimes\sigma_3$ }] (C') at (3,0,-3);

\draw (A)--(V);
\draw (B)--(V);
\draw (C)--(V);

\draw (A)--(A');
\draw (B)--(B');
\draw (C)--(C');

\draw [line width =2pt] (A)--(B);
\draw [line width =2pt] (B)--(C);
\draw [line width =2pt] (C)--(A);

\draw [line width =2pt] (A')--(B');
\draw [line width =2pt] (B')--(C');
\draw [line width =2pt]  (C')--(A');

\node [label= 45:$I\otimes\sigma_2\sigma_1$]   (F) at (intersection of A--B and A'--B') {};

\draw [dashed] (A)--(F); \draw [dashed] (A')--(F);

\node [label= 45:$I\otimes\sigma_3\sigma_2$]   (G) at (intersection of B--C and B'--C') {};

\draw [dashed] (C)--(G); \draw [dashed] (C')--(G);

 \node [label= 45:$I\otimes\sigma_3\sigma_1$]   (H) at (intersection of A--C and A'--C') {};

 \draw [dashed] (C)--(H); \draw [dashed] (C')--(H);

 \draw [dashed] (F)--(G); \draw [dashed] (G)--(H);
 \foreach \point in {F,G,H}
\fill [black,opacity=.5] (\point) circle (2pt);
  \end{tikzpicture}
\caption{The Desargues Configuration for ${\clf}_3\otimes{\clf}_3$}
\label{FigDesPauli}
\end{figure}
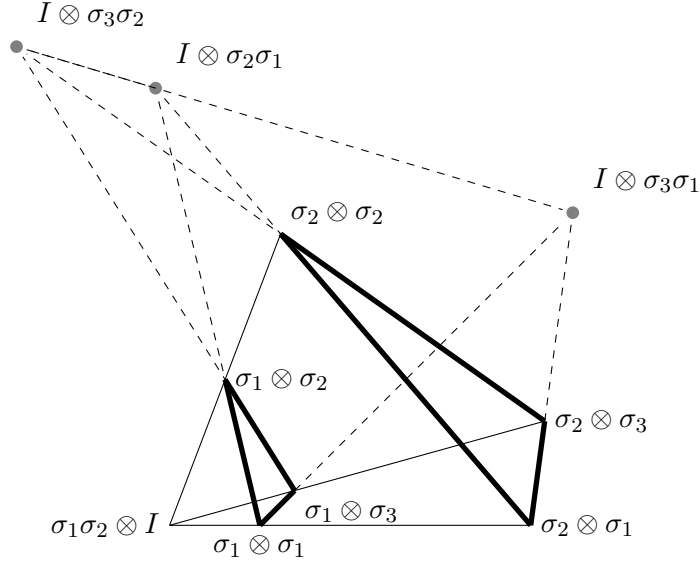

We can now state our results formally:
\begin{prop}\label{P:prjvec} Let $\mbp_m$ be the set of all $e_S\in\clf_m$ with non-empty $S\subset\{1,...,m\}$, and let ${\mathbb L}_m$
be the set of all triples of distinct elements $a,b,c\in \mbp_m$ such that the product of two of them is $\pm$ the third. The finite geometry $(\mbp_m,{\mathbb L}_m)$
 for which
 $\mbp_m$ is the set of
points  and ${\mathbb L}_m$   the set of lines is a projective geometry isomorphic to the  geometry of the projective space $\mbp(\mbz_2^{m})$.

Moreover, if $X$ is the set of points in a projective subspace of $(\mbp_m,{\mathbb L}_m)$ then the elements $e_S$ which are in $X$ span
a Lie algebra under the bracket commutator.
\end{prop}

\begin{proof} The first statement in the conclusion has already been proven in the preceding discussion. Now let $X$ be the set of points in a projective subspace of $(\mbp_m,{\mathbb L}_m)$. If $e_S$ and $e_T$ are in $X$ then the third point on the line through these points is $\pm e_Se_T$. Now the commutator bracket of
$e_S$ and $e_T$ in $\clf_m$ is
$$[e_S,e_T]=e_Se_T -e_Te_S=(1-\epsilon_{S,T}\epsilon_{T,S})e_Se_T.$$
Thus, the bracket is in the linear span of the elements $e_R\in X$. Therefore, the linear span of $X$ is a Lie algebra. 
\end{proof}

Rau \cite{Rau} shows that certain sets of points  give rise to Lie algebras even though they do not form
a projective subspace, but rather are Desargues configurations. It is also interesting to note, as Rau points out,  that it is possible to
decorate the lines of the Fano plane with arrows and,
with this decoration, the Fano plane encodes the multiplication of octonions. Since octonion multiplication is not associative, this
structure is not the same as that of the Clifford algebra $\clf_3$. 

After completing the first version of this paper I came across the papers by Shaw \cite{Sh} and Shaw and Jarvis \cite{ShJ} which explored in detail the
geometries   arising from irreducible representations of a certain class of Clifford algebras.  

{\em Acknowledgements}. My thanks to Bill Hoffman, Ravi Rau, and Neal Stoltzfus for useful discussions. Figures were drawn with Till Tantau's Ti{\em k}Z.

\bibliographystyle{amsplain}

\end{document}